\documentclass{amsart}

\usepackage[centertags]{amsmath}
\usepackage{amsfonts}
\usepackage{amssymb}
\usepackage{amsthm}
\usepackage{euscript}
\usepackage[colorlinks=true]{hyperref}
\newtheorem{theorem}{Theorem}
\newtheorem{prop}{Proposition}

\newtheorem{corollary}[prop]{Corollary}

\theoremstyle{definition}

\theoremstyle{remark}

\long\def\symbolfootnote[#1]#2{\begingroup%
\def\thefootnote{\fnsymbol{footnote}}\footnote[#1]{#2}\endgroup}

\newcommand{\nablash}{\nabla{\kern -.75 em
     \raise 1.5 true pt\hbox{{\bf/}}}\kern +.1 em}
\newcommand{\Deltash}{\Delta{\kern -.69 em
     \raise .2 true pt\hbox{{\bf/}}}\kern +.1 em}
\newcommand{\Rslash}{R{\kern -.60 em
     \raise 1.5 true pt\hbox{{\bf/}}}\kern +.1 em}
\newcommand{\Tr}{\operatorname{Tr}}

\newcommand{\D}{\partial}

\renewcommand{\div}{\operatorname{div}}

\newcommand{\RN}{Reissner-Nordstr\"om}

\begin{document}
\title[The Riemannian Penrose inequality with charge]{On the Riemannian Penrose inequality with charge and the cosmic
censorship conjecture}
\subjclass[2000]{Primary 83C57, Secondary 58J05}

\author{Marcus A.~Khuri, Gilbert Weinstein, and Sumio Yamada}
\address{Department of Mathematics\\ Stony Brook University\\ Stony Brook, NY 11794}
\email{khuri@math.sunysb.edu}
\address{Physics Dept. and Dept.\ of  Computer Science and Mathematics\\ Ariel University of Samaria\\ Ariel 40700,
Israel}
\email{gilbert.weinstein@gmail.com}
\address{Department of Mathematics\\Gakushuin University\\ Tokyo 171-8588, Japan}
\email{yamada@math.gakushuin.ac.jp}
\thanks{M.~A.~Khuri was supported by NSF grant DMS-1007156.}
\thanks{Part of this paper was discussed while the authors were attending the RIMS workshop \emph{Geometry of
Moduli Spaces for Low Dimensional Manifolds} at Kyoto University in October 2012. The authors would like to thank the
RIMS for its hospitality.}

\begin{abstract}
We note an area-charge inequality originally due to Gibbons: if the outermost horizon $S$ in an asymptotically flat
electrovacuum initial data set is connected then $|q|\leq r$, where $q$ is the total charge and $r=\sqrt{A/4\pi}$ is
the area radius of $S$. A consequence of this inequality, in conjunction with the positive mass theorem with charge, is
that for connected black holes the following lower bound on
the area holds: $r\geq m-\sqrt{m^2-q^2}$. When combined with the upper bound $r\leq m + \sqrt{m^2-q^2}$ which is
expected to hold always, this implies the natural generalization of the Riemannian Penrose inequality: $m\geq
1/2(r+q^2/r)$. We also establish the same lower bound without the assumption of time symmetry.
\end{abstract}

\maketitle






A natural generalization of the Riemannian Penrose inequality incorporating electric charge is:
\begin{equation} 	\label{CRPI}
  m \geq \frac12\left(r + \frac{q^2}{r}\right),
\end{equation}
with equality if and only if the data is \RN. Here $m$ is the ADM mass, $r=\sqrt{A/16\pi}$ is
the area radius of the outermost horizon, and $q$ is the total charge. This inequality is
known to hold when the outermost horizon is connected, but could be violated otherwise. The proof of
the inequality for a connected horizon follows the inverse mean curvature flow argument of
Huisken-Ilmanen~\cite{huisken-ilmanen}, with a minor modification based on an earlier observation of Jang~\cite{jang}.
The rigidity statement was recently proved by Khuri and Disconzi~\cite{disconzi-khuri}.
The inequality~\eqref{CRPI} can fail when the outermost
horizon is not connected, as shown in~\cite{weinstein-yamada}. Indeed, the area of the cross-section of the necks in a
Majumdar-Papapetrou solution with two bodies violates~\eqref{CRPI}, and while this solution is not asymptotically flat,
and does not contained a compact minimal surface, these deficiencies were corrected in~\cite{weinstein-yamada} by gluing
using a perturbation argument two identical copies along the necks. As observed already by Jang, inequality~\eqref{CRPI}
is equivalent to the two inequalities:
\begin{equation} \label{upper-lower-bound}
  m-\sqrt{m^2-q^2}\leq r \leq m+\sqrt{m^2-q^2}.
\end{equation}
We note that $|q|\leq m$ follows from the positive mass theorem with charge~\cite{GHHP}.

The upper bound on $r$ in~\eqref{upper-lower-bound} is suggested by cosmic censorship, via a heuristic argument of
Penrose.  If the data violates this inequality, and if the evolution is smooth enough, then one expects the area of the
horizon to be non-decreasing while the other parameters are constant. Hence the inequality will also be violated in
the limit of late times, in contradiction to all the known stationary solutions without naked
singularities. In a future paper, we prove this upper bound using a modification of the conformal flow used by H.~Bray.
An outline of the proof will appear in~\cite{khuri-weinstein-yamada}.

The counter-example constructed in~\cite{weinstein-yamada} violates the lower bound in~\eqref{upper-lower-bound}, and
one might say that this is not so surprising  since there is no physical motivation for
this lower bound. On the other hand, as Robert Wald~\cite{wald} pointed out, the fact that the lower
bound is always satisfied when the outermost horizon is connected might be surprising, and indeed no physical
motivation has so far been proposed for this inequality. In this short note, we show that in this case,
the lower bound in fact follows from an inequality
first proved by Gibbons~\cite{gibbons} using the stability of the outermost horizon as an area minimizing surface.
Thus, while Penrose's heuristic argument based on cosmic censorship provides a physical
justification for the upper bound on $r$, our Corollary~1 below shows that the positive mass theorem with charge  in
conjunction with Gibbons' inequality provides a physical justification for the lower
bound on $r$ when the outermost horizon is connected.

We begin by very briefly introducing a few definitions. An initial data set $(M,g,E)$
consists of a $3$-manifold
$M$, a Riemannian metric $g$, and a vector field $E$.
\symbolfootnote[2]{For simplicity, we first assume that the magnetic field vanishes.}
We assume that the data satisfies the
Maxwell constraint $\div_g E=0$, and
the dominant energy condition $R\geq 2|E|^2$, where $R$ is the scalar curvature of $g$. We assume that the data is
strongly asymptotically flat meaning that there is a compact set $K$ such that $M\setminus K$ is the finite union of
disjoint ends, and on each end the fields decay according to:
\[
  g-\delta = O(|x|^{-1}), \qquad E = O(|x|^{-2}).
\]
In addition we assume that $R_g$ is integrable. This guarantees that the ADM mass and total charge
\[
  m = \frac1{16\pi} \int_{S_\infty} (g_{ij,j}-g_{jj,i}) \nu^i\, dA, \qquad
  q = \frac1{4\pi} \int_{S_\infty} E_i \nu^i\, dA
\]
are well defined. Here, $\nu$ is the outer unit normal, and the limit is taken in a designated end. Conformally
compactifying all but the designated end, we
can now restrict our attention to surfaces which bound compact regions, and define $S_2$ to enclose $S_1$ to mean
$S_1=\D K_1$, $S_2=\D K_2$ and $K_1\subset K_2$. An \emph{outermost horizon} is a compact minimal surface not enclosed
in any other compact minimal surface.

A version of the following theorem was proved in~\cite[Section 6]{gibbons}. The main difference is that instead of
stability we assume the outermost condition, which then in turn implies stability. The outermost condition is a natural
condition appearing in statements of the Penrose inequality. We bring it here for completeness and because the proof is
very simple.

\begin{theorem} \label{area-charge}
Let $(M,g,E)$ be strongly asymptotically flat, satisfying $R\geq 2|E|^2$ and $\div_g E=0$, and suppose the outermost
horizon $S$ is connected. Then
\[
  |q| \leq r.
\]
\end{theorem}

\begin{proof}
We begin by pointing out that $S$ is in fact outer minimizing, meaning
that it has area no greater than any other surface which encloses it, see for example~\cite[pg
706]{weinstein-yamada}. Thus, $S$ is a stable minimal surface and from the second variation of area we get:
\[
  0\leq\int_S-(|\chi|^2+R_{\nu\nu})\, dA =\int_S \left( \kappa - \frac{1}{2}|\chi|^{2}-\frac{1}{2}R \right)\, dA,
\]
where $\chi$ is the second fundamental form of $S$ in $M$, $\nu$ is the outward unit normal, and $\kappa$ is the Gauss
curvature of $S$. The second equality follows from the Gauss equation. Since $S$ is connected, it follows from
Gauss-Bonnet that $\int_S \kappa\, dA = 4\pi$ hence
\[
    4\pi \geq \int_S \frac{1}{2}R\,dA \geq \int_S |E\cdot\nu|^2\, dA
    \geq \frac1{4\pi r^2} \left( \int_S E\cdot\nu\,dA \right)^2
    = \frac{(4\pi q)^2}{4\pi r^2}.
\]
\end{proof}

\begin{corollary}
Under the same hypotheses as Theorem~\ref{area-charge}, we have:
\[
  r \geq m - \sqrt{m^2-q^2}.
\]
\end{corollary}

\begin{proof}
\begin{equation}\label{3}
  m = \sqrt{q^2+m^2-q^2} \leq |q| + \sqrt{m^2-q^2} \leq r + \sqrt{m^2-q^2}.
\end{equation}
This proves the lower bound on $r$.
\end{proof}

Once the lower bound and the upper bound $r\leq m+\sqrt{m^2-q^2}$ both hold, inequality~\eqref{CRPI} follows. However,
we point out that in fact we do not know of any independent proof of this upper bound.

We end by generalizing these observations in two ways, namely by including the magnetic field and
removing the assumption of time symmetry. In this regard, consider the general initial data set $(M, g, k, E, B)$, where
$k$ is a symmetric 2-tensor representing the extrinsic curvature of $M$ in spacetime, and $B$ is a vector field
representing the magnetic field. For strong asymptotic flatness, we require these quantities to
satisfy the following fall-off conditions at spatial infinity
\[
  k=o(|x|^{-2}), \qquad
  B=O(|x|^{-2}).
\]
Similarly to the above, the total charges are now given by
\[
  q_{e} = \frac1{4\pi} \int_{S_\infty} E_i \nu^i\, dA, \qquad
  q_{b} = \frac1{4\pi} \int_{S_\infty} B_i \nu^i\, dA,
\]
and the matter and current densities for the non-electromagnetic matter fields are given by
\begin{align*}
 2 \mu  & = R + (\Tr k)^2 - |k|^2 - 2(|E|^2+|B|^{2}), \\
 J & = \div_{g} (k - (\Tr k)g)+2E\times B.
\end{align*}
The following is a Corollary of Theorem~2.1 in~\cite{dain}.

\begin{corollary} \label{area-charge1}
Let $(M,g,k,E,B)$ be strongly asymptotically flat, satisfying $\mu\geq|J|$ and $\div_g E=\div_g B=0$, and suppose the outermost
apparent horizon $S$ is connected. Then
\begin{equation} \label{area-charge-non-sym}
  r \geq m-\sqrt{m^2-q_e^2-q_b^{2}}.
\end{equation}
\end{corollary}

\begin{proof}
As $S$ is outermost, it is stable, and hence Theorem 2.1 in~\cite{dain} implies the area-charge inequality
\[
  \sqrt{q_e^2+q_b^2} \leq r.
\]
Under the current assumptions the positive mass theorem with charge yields
\[
  m\geq\sqrt{q_e^2+q_b^2}.
\]
These two inequalities combine as in \eqref{3} to give
\[
 m \leq \sqrt{q_e^2+q_b^2} + \sqrt{m-q_e^2-q_b^2} \leq r + \sqrt{m-q_e^2-q_b^2}
\]
and~\eqref{area-charge-non-sym} follows.
\end{proof}

\section*{Acknowledgement}
The authors would like to thank Sergio Dain, Jos\'{e} Luis Jaramillo, Marc Mars, and Mart\'{\i}n Reiris for their comments and
suggestions.

\bibliographystyle{amsplain}

\end{document}